\documentclass[10pt,twocolumn]{article}
\usepackage[letterpaper,margin=0.72in,columnsep=0.25in]{geometry}
\usepackage[T1]{fontenc}
\usepackage{lmodern}
\usepackage{amsmath,amssymb,amsthm,mathtools,bm}
\usepackage{abstract}
\usepackage{graphicx,booktabs,array}
\usepackage[dvipsnames]{xcolor}
\usepackage{tikz}
\usetikzlibrary{arrows.meta,positioning,calc}
\usepackage{microtype}
\usepackage[numbers,sort&compress]{natbib}
\usepackage[colorlinks=true,linkcolor=NavyBlue,citecolor=NavyBlue,urlcolor=NavyBlue]{hyperref}
\definecolor{flowblue}{HTML}{285F91}
\definecolor{flowgreen}{HTML}{287A68}
\definecolor{flowred}{HTML}{AF3F43}
\definecolor{flowgray}{HTML}{52616C}

\newcommand{\NS}{\mathcal{NS}}

\newcommand{\R}{\mathbb{R}}
\newcommand{\one}{\mathbf{1}}
\newcommand{\PC}{\mathcal{P}}
\newcommand{\CC}{\mathcal{C}}
\newcommand{\U}{P_{\mathrm{unif}}}

\newtheorem{theorem}{Theorem}
\newtheorem{proposition}[theorem]{Proposition}
\newtheorem{corollary}[theorem]{Corollary}
\newtheorem{lemma}[theorem]{Lemma}
\theoremstyle{definition}
\newtheorem{definition}[theorem]{Definition}
\newtheorem{remark}[theorem]{Remark}

\tikzset{box/.style={draw=flowblue,rounded corners=2pt,fill=flowblue!5,
align=center,inner sep=5pt,font=\small},
arr/.style={-{Stealth[length=2mm]},thick,draw=flowgray}}
\title{Certifying information flow in multi-agent behaviours}
\author{%
M.~Haider Akbar\textsuperscript{1}\thanks{makbar25@ku.edu.tr}
\qquad
Onur Pusuluk\textsuperscript{2}\thanks{onur.pusuluk@khas.edu.tr}
\qquad
\"Ozg\"ur E.~M\"ustecapl\i o\u{g}lu\textsuperscript{3,4}\thanks{omustecap@ku.edu.tr}\\[8pt]
{\small\textsuperscript{1}Department of Computational Sciences and Engineering,
Ko\c{c} University, 34450 Sar\i yer, Istanbul, T\"urkiye}\\
{\small\textsuperscript{2}Faculty of Engineering and Natural Sciences,
Kadir Has University, Istanbul 34083, Turkey}\\
{\small\textsuperscript{3}Department of Physics, Ko\c{c} University,
34450 Sar\i yer, Istanbul, T\"urkiye}\\
{\small\textsuperscript{4}T\"UB\.ITAK Research Institute for Fundamental
Sciences (TBAE), 41470 Gebze, T\"urkiye}%
}
\date{}

\begin{document}
\raggedbottom

\twocolumn[
\begin{@twocolumnfalse}
\begin{center}
\end{center}
\vspace{-1em}
\maketitle
\begin{abstract}
We study which input--output behaviours of several agents can be produced by a declared communication schedule, that is, a fixed sequence of senders, recipients, message alphabets, and output times, together with shared randomness. Since a schedule admits only finitely many deterministic protocols, membership is a linear program, and an infeasible target is excluded by a separating inequality. Under explicit operational assumptions, we prove that the affine hull of the achievable behaviours is described exactly by the equalities stating that each joint output statistic depends only on the inputs that can reach the agents involved, so that a violated equality identifies both the statistic and the missing input. One binary message per event already attains this affine hull, so message capacity restricts which behaviours are achievable without being visible to any equality. We apply the result to one-bit broadcast and relay hierarchies for two and three agents, and to low-order projections of the local polytope.
\end{abstract}
\vspace{1.5em}
\end{@twocolumnfalse}
]
\saythanks
\section{Introduction}
\label{sec:intro}

Correlations do not by themselves reveal how information was exchanged. Several agents can respond coherently because they communicate during a task, because they share information established beforehand, or both. Bell scenarios make this distinction precise: deterministic local response tables are the vertices of a polytope, shared randomness generates their convex hull, and Bell inequalities exclude behaviours outside it~\cite{CHSH1969,Brunner2014}. The same convex approach has been extended to bounded classical communication~\cite{BaconToner2003,Maxwell2014,Cruzeiro2019}, to lower bounds on average communication cost~\cite{Pironio2003}, and to more general causal correlations~\cite{Abbott2016}.

A separating inequality shows that a target is incompatible with a resource model, but it need not indicate which information pathway is missing. Affine equalities carry complementary information. The no-signalling equalities state that a subsystem's statistics do not depend on a remote input, and the local and no-signalling polytopes have the same affine hull although their inequalities differ~\cite{Pironio2005}. An equality can therefore diagnose a forbidden dependence, whereas satisfying every equality does not establish membership.

In this work the resource is a fixed communication schedule, specified by senders, recipients, message alphabets, and output times. Under explicit operational assumptions, we show that the affine hull of the behaviours achievable with a schedule is determined by reachability: a joint statistic of the outputs of a set of agents can depend only on inputs that reach at least one of those agents, and every such dependence is attained in the affine sense (Theorem~\ref{thm:affine}). In this work the resource is a fixed communication schedule, specified by senders, recipients, message alphabets, and output times. Under explicit operational assumptions, we show that the affine hull of the behaviours achievable with a schedule is determined by reachability: a joint statistic of the outputs of a set of agents can depend only on inputs that reach at least one of those agents, and every such dependence is attained in the affine sense (Theorem~\ref{thm:affine}). In particular, one binary message per event already saturates the affine hull, so no affine equality can detect a shortage of message capacity. Equalities and inequalities therefore answer different questions: a violated equality identifies which input failed to reach which agents, whereas only an inequality can show that a channel was too narrow.

The same moment coordinates show what is lost under low-order projections. We show that every projection of the local polytope onto moments of order at most $K$ contains an explicit $\ell_1$ ball around the uniform behaviour. For GHZ states, the proper reduced states coincide with those of a separable mixture, a known obstruction to few-body detection~\cite{Tura2015,Tan2016} that is established for all proper orders in Ref.~\cite[Sec.~VI]{Tan2016}. For $\sigma_x/\sigma_y$ settings we give an explicit moment-space proof of this obstruction together with a quantitative margin. This complements the detection of other many-body states with two-body Bell inequalities~\cite{Tura2014,Tura2015}.

Finally, we consider agents and observers. Quantum decision theory models choices with Hilbert-space structures and with utility and attraction factors~\cite{YS2010,Favre2016,Yukalov2020}, and network extensions include explicit information exchange~\cite{Yukalov2018}; recent formulations retain the utility attraction decomposition in classical probabilistic terms and model repeated exchange and imitation among groups of agents~\cite{YY2026a,YY2026b}. We show that an interference term in an individual response rule does not produce nonlocality when each agent uses only its own input and shared prior information, and that a marginal-preserving joint correction changes only the dependence between outputs. Extended Wigner's friend arguments~\cite{Brukner2018,FR2018,Bong2020} also declare which agents measure, which records are passed to whom, and when. We locate them in the present framework: Local Friendliness violations lie in the no-signalling affine hull and can therefore be certified only by inequalities, one message carrying a setting suffices to reproduce them, and in tripartite scenarios where Local Friendliness coincides with locality the GHZ results apply directly. The Frauchiger--Renner ~\cite{FR2018} argument involves no setting choices and lies outside the behaviour-level certificates considered here. These are conditional mathematical statements about input-output behaviours; they do not assess decision theory or interpretations of quantum theory empirically.

Section~\ref{sec:moments} introduces behaviours and moment coordinates, and Table~\ref{tab:notation} collects the principal notation.
Section~\ref{sec:affine} states the main result and the associated certificates,
Sec.~\ref{sec:comm} develops declared hierarchies for two and three agents,
Sec.~\ref{sec:loss} treats low-order projections, Sec.~\ref{sec:agents} discusses decision rules and extended Wigner's friend scenarios, and
Sec.~\ref{sec:disc} relates the results to prior work.
Appendix~\ref{app:example} applies these results to a single three-agent selection task, contrasting equality and inequality certificates.

\section{Behaviours and moments}
\label{sec:moments}

There are $N$ agents, indexed by $[N]=\{1,\ldots,N\}$. Agent $i$ receives an input $x_i$ from a finite set $X_i$ of size $d_i$ and returns $a_i\in\{0,1\}$; its signed output is $A_i=(-1)^{a_i}$. A behaviour is a family
\begin{equation}
 P(a\mid x)\ge0,\qquad \sum_a P(a\mid x)=1
 \label{eq:behaviour}
\end{equation}
on the complete Cartesian input set $X=\prod_iX_i$. The shared random variable $\lambda$, with distribution $\mu$, is independent of $x$. Mutual independence of the inputs is not needed; input coverage and independence from $\lambda$ are the relevant assumptions. The uniform behaviour $\U$ has $P(a\mid x)=2^{-N}$.

In the absence of communication,
\begin{equation}
 P(a\mid x)=\int\prod_i p_i(a_i\mid x_i,\lambda)\,d\mu(\lambda).
 \label{eq:local}
\end{equation}
Absorbing all private randomness into $\lambda$, each agent follows a deterministic table $a_i=\alpha_i(x_i)$. The local polytope $\PC_0$ is the convex hull of the $2^{\sum_i d_i}$ resulting deterministic behaviours, each of which is extreme. For two agents with two binary inputs this is equivalent to Fine's joint distribution for the four potential outcomes~\cite{Fine1982}.

For nonempty $S\subseteq[N]$, the moment
\begin{equation}
 E_S(x)=\sum_a\Bigl[\prod_{i\in S}(-1)^{a_i}\Bigr]P(a\mid x)
 \label{eq:moment}
\end{equation}
is the expected product of the signed outputs in $S$ at context $x$; when agents
have names we write $E_A$, $E_{AB}$, and so on. Fourier inversion on the binary
output cube \cite{ODonnell2014} gives, at every context,
\begin{equation}
 P(a\mid x)=2^{-N}\Bigl[1+\sum_{S\ne\emptyset}
 \Bigl(\prod_{i\in S}(-1)^{a_i}\Bigr)E_S(x)\Bigr],
 \label{eq:inversion}
\end{equation}
with or without communication. A behaviour is no-signalling if each $E_S$
depends only on $x_S$; equivalently, each subsystem marginal ignores the other
agents' inputs. Local behaviours are no-signalling, and $\U$ has all moments
equal to zero.

If each moment $E_S$ may depend only on the inputs of a set $R(S)$ of agents, then listing each $E_S$ once as a function of $x_{R(S)}$, and keeping moments of order at most $K$, gives
\begin{equation}
 D_R(K)=\sum_{\substack{S\ne\emptyset\\|S|\le K}}\ \prod_{j\in R(S)}d_j
 \label{eq:count}
\end{equation}
coordinates. With $R(S)=[N]$ for every $S$, $D_{\mathrm{full}}:=D_R(N)=(2^N-1)\prod_id_i$ is the dimension of the set of all behaviours. With the no-signalling restriction $R(S)=S$ we write $D_0(K)$; then $D_0(N)=\prod_i(d_i+1)-1$, and for equal $d_i=d$, $D_0(K)=\sum_{m=1}^K\binom Nm d^m$, so that $D_0(2)=2N^2$ for $d=2$.

The projection $\Pi_K$ keeps only the moments of order at most $K$, and $\CC_K=\Pi_K(\PC_0)$ is the image of the local polytope.

\begin{proposition}
\label{prop:localdim}
$\dim\CC_K=D_0(K)$.
\end{proposition}
\begin{proof}
By Ref.~\cite[Theorem~1]{Pironio2005}, the local polytope has the full no-signalling dimension $D_0(N)$ and is therefore full-dimensional in the coordinates $E_S(x_S)$; so is its projection onto the $D_0(K)$ retained coordinates.
\end{proof}

For three agents with binary inputs, $\Pi_2$ retains $18$ of the $26$ independent coordinates. The following definition and lemma make Proposition~\ref{prop:localdim} quantitative.

\begin{definition}
\label{def:basis}
For a nonempty set $S$ of agents and an input assignment $r\in X_S$, let $\bm e_{S,r}$ denote the moment vector with $E_S(x_S)=\one\{x_S=r\}$ and $E_T(x_T)=0$ for every $T\ne S$; that is, $\bm e_{S,r}$ is the standard basis vector along the coordinate $E_S(r)$.
\end{definition}

For example, for two agents with binary inputs the coordinates are $E_A(x)$, $E_B(y)$, and $E_{AB}(x,y)$, eight in total, and $\bm e_{AB,(1,0)}$ has
$E_{AB}(1,0)=1$ and all other coordinates equal to zero. In the corresponding behaviour each output is uniformly random; the two outputs always agree at inputs $(x,y)=(1,0)$ and are uncorrelated at every other input pair.

\begin{lemma}
\label{lem:l1}
For every nonempty $S$, every $r\in X_S$, and both signs, $\pm\bm e_{S,r}$ (Definition~\ref{def:basis}) is the moment vector of a local behaviour. Consequently, for every $K$,
\begin{equation}
 \{u\in\R^{D_0(K)}:\|u\|_1\le1\}\subseteq\CC_K ,
 \label{eq:ball}
\end{equation} 
where $\|u\|_1=\sum_j|u_j|$. Moreover, let $\bm c\in\R^{D_0(K)}$ and $\beta\in\R$ define a valid inequality for $\CC_K$, that is, $\bm c^\top u\le\beta$ for all $u\in\CC_K$. If $\bm c\ne0$, then $\beta\ge\|\bm c\|_\infty=\max_j|c_j|>0$.
\end{lemma}

\begin{proof}
Fix $S$, $r$, and a sign $\tau\in\{\pm1\}$. Let the shared variable consist of signs $(\xi_i)_{i\in S}$, uniformly distributed on the set $\{\xi:\prod_{i\in S}\xi_i=\tau\}$, together with independent unbiased signs $\xi'_1,\ldots,\xi'_N$. Agent $i\in S$ outputs $\xi_i$ if $x_i=r_i$ and $\xi'_i$ otherwise, and agent $i\notin S$ outputs $\xi'_i$. Each output depends only on the agent's own input and the shared variable, so the behaviour is local.

Let $T\ne\emptyset$. If $T\not\subseteq S$, or if $x_i\ne r_i$ for some $i\in T$, the product of the outputs in $T$ contains a factor $\xi'_i$ that is unbiased and independent of the other factors, so $E_T=0$. If $T\subseteq S$ and $x_T=r_T$, then $E_T=\mathbb E\prod_{i\in T}\xi_i$, which equals $\tau$ for $T=S$. For $\emptyset\ne T\subsetneq S$, choose $i\in T$ and $j\in S\setminus T$; flipping $\xi_i$ and $\xi_j$ together preserves the distribution of $\xi$ and reverses the sign of $\prod_{k\in T}\xi_k$, so $E_T=0$. The moment vector is therefore $\tau\bm e_{S,r}$.

The vectors $\pm\bm e_{S,r}$ are the signed standard basis vectors of $\R^{D_0(N)}$, and their convex hull is the unit $\ell_1$ ball. Since the local polytope is convex, it contains this ball, and $\Pi_K$ maps it onto the unit $\ell_1$ ball of $\R^{D_0(K)}$, which proves Eq.~\eqref{eq:ball}. Finally, let $(S,r)$ with $|S|\le K$ index a coordinate at which $|c_{S,r}|=\|\bm c\|_\infty$. The point $u=\operatorname{sign}(c_{S,r})\,\Pi_K(\bm e_{S,r})$ lies in the ball and hence in $\CC_K$, and $\bm c^\top u=\|\bm c\|_\infty$, so $\beta\ge\|\bm c\|_\infty$.
\end{proof}

\begin{table}[tb]
\centering
\caption{Principal notation.}
\label{tab:notation}
\small
\begin{tabular}{@{}p{2.4cm}p{5.6cm}@{}}
\toprule
Symbol & Meaning\\
\midrule
$x_i\in X_i$, $d_i$ & Input of agent $i$; alphabet; $d_i=|X_i|$\\
$a_i$, $A_i$ & Output; signed output $(-1)^{a_i}$\\
$P(a\mid x)$, $\U$ & Behaviour; uniform behaviour\\
$\lambda$, $\mu$ & Shared randomness; its distribution or convex weights\\
$E_S(x)$ & Moment of the agents in $S$, Eq.~\eqref{eq:moment}\\
$D_R(K)$ & Number of moment coordinates, Eq.~\eqref{eq:count}\\
$\Pi_K$, $\CC_K$ & Projection onto moments of order $\le K$; $\CC_K=\Pi_K(\PC_0)$\\ $\mathcal A$, $\PC_{\mathcal A}$, $\PC_0$ & Schedule; its class; local class\\ $R_{\mathcal A}(S)$ & Inputs that can reach the outputs in $S$\\ $\bm m$; $(\bm c,\beta)$; $\bm w$ & Target moments; separating inequality; equality certificate\\
\bottomrule
\end{tabular}
\end{table}

\section{Information flow, affine hulls, and certificates}
\label{sec:affine}

\paragraph{A two-agent example.}
Take binary inputs $x$ for Alice and $y$ for Bob. Alice sends one message $m_1$, a function of $x$, to Bob before he responds, and Alice receives nothing. Conditional on $\lambda$, Alice's output is a function of $x$, while Bob's is a function of $y$ and $m_1$. Hence $E_A$ cannot depend on $y$, whereas $E_B$ and $E_{AB}$ may depend on both inputs. Of the $3\times4=12$ moments, the two equalities
\begin{equation}
 E_A(x,1)-E_A(x,0)=0,\qquad x\in\{0,1\},
 \label{eq:example}
\end{equation}
are forced, leaving $10$ free coordinates, and Theorem~\ref{thm:affine} below shows that there are no further affine constraints. The target in which Alice outputs $y$ and Bob outputs $0$ has $E_A(0,1)-E_A(0,0)=-2$ and so violates Eq.~\eqref{eq:example} at $x=0$. The violation identifies the statistic (Alice's marginal) and the input ($y$) that the schedule does not deliver to Alice; a separating inequality would also exclude this target but need not point to $y$. In this example the equalities together with positivity describe the class exactly (Proposition~\ref{prop:two}); this is not true in general, as shown at the end of this section.

\begin{definition}
\label{def:arch}
A schedule $\mathcal A$ fixes a finite sequence of message events, each with a sender, a set of recipients, and a finite alphabet, together with the time at which each agent's output is produced. Inputs arrive before any message. Agents remember their input and everything they have received, and may apply any deterministic function to this history. Shared randomness is independent of the inputs. The class $\PC_{\mathcal A}$ is the convex hull of the deterministic behaviours that the schedule can produce.
\end{definition}

Definition~\ref{def:arch} specifies the schedule but leaves implicit several properties of the protocol family on which the main result depends. We state them separately, because each enters a specific step of the proof of Theorem~\ref{thm:affine} below and each corresponds to a modelling choice that one may wish to vary.
\begin{itemize}
\item[(i)] Outputs are final: no message uses a produced output as an argument, so flipping any agent's output, with all messages unchanged, gives another protocol of the same schedule. This is used to isolate individual moments.
\item[(ii)] Reachability is structural: which inputs can reach which outputs is read off the schedule, not off the functions a particular protocol uses, and a protocol may ignore what it receives. This makes the equalities depend on the schedule alone.
\item[(iii)] Messages are free: each message may be any function of the sender's history, and each alphabet has at least two noiseless symbols that the sender can choose between. This is used to transmit the one-bit flags in the spanning argument.
\item[(iv)] Timing and routing are fixed: who speaks, to whom, and when do not depend on the inputs, directly or through earlier messages, so the presence or absence of a message carries no information.
\end{itemize}
All four properties hold for the schedules of Sec.~\ref{sec:comm}. They concern the affine hull only; two schedules with the same reachability have the same affine hull but can have different polytopes.

Each schedule $\mathcal A$ determines a directed acyclic graph, its time-ordered information graph, with one node for each input, message event, and output, and with edges representing memory and message delivery (Fig.~\ref{fig:reach}). Input $x_j$ reaches output $i$ if this graph contains a directed path from $x_j$ to output $i$. We denote by $R_{\mathcal A}(\{i\})$ the set of such $j$ and set $R_{\mathcal A}(S)=\bigcup_{i\in S}R_{\mathcal A}(\{i\})$. A message that arrives after an agent's output has been fixed does not reach that output. We write $D_{\mathcal A}(K)$ for Eq.~\eqref{eq:count} with $R=R_{\mathcal A}$.

For example, consider three agents with inputs $x$, $y$, $z$ and the relay schedule of Fig.~\ref{fig:reach}: Alice sends $m_1(x)$ to Bob, Bob then sends $m_2(y,m_1)$ to Charlie, and all outputs are produced after these messages. Labelling inputs by name, Alice's output is reached only by $x$, Bob's by $\{x,y\}$ through $m_1$, and Charlie's by $\{x,y,z\}$ through the chain $x\to m_1\to m_2$. Reach sets of groups are unions, so $R_{\mathcal A}(\{A,C\})=\{x,y,z\}$: the joint statistic $E_{AC}$ may depend on all three inputs, although $E_A$ may depend only on $x$. Timing matters: if Bob produced his output before $m_1$ arrived, his reach would shrink to $\{y\}$. Summing $2^{|R(S)|}$ over the seven sets $S$ gives $D_{\mathcal A}(3)=2+4+8+4+8+8+8=42$, the value listed for this schedule in Table~\ref{tab:reach}.

\begin{figure}[tb]
\centering
\caption{Time-unrolled two-bit relay of Sec.~\ref{sec:three}, with $m_1=m_1(x)$ and $m_2=m_2(y,m_1)$. Grey lines are memory and red arrows are messages. Each output's reach is shown on the right; the reach of a set of agents is the union, so $E_{AC}$ may depend on $\{x,y,z\}$ while $E_A$ may depend only on $x$.}
\begin{tikzpicture}[font=\small]
 \node[box,minimum width=0.6cm] (ix) at (0,0) {$x$};
 \node[box,minimum width=0.6cm] (iy) at (0,-1.15) {$y$};
 \node[box,minimum width=0.6cm] (iz) at (0,-2.3) {$z$};
 \node[box,minimum width=0.6cm,draw=flowgreen,fill=flowgreen!5] (oa) at (4.4,0) {$a$};
 \node[box,minimum width=0.6cm,draw=flowgreen,fill=flowgreen!5] (ob) at (4.4,-1.15) {$b$};
 \node[box,minimum width=0.6cm,draw=flowgreen,fill=flowgreen!5] (oc) at (4.4,-2.3) {$c$};
 \node[anchor=east,font=\footnotesize] at (-0.4,0) {$A$};
 \node[anchor=east,font=\footnotesize] at (-0.4,-1.15) {$B$};
 \node[anchor=east,font=\footnotesize] at (-0.4,-2.3) {$C$};
 \draw[thick,flowgray] (ix.east)--(oa.west);
 \draw[thick,flowgray] (iy.east)--(ob.west);
 \draw[thick,flowgray] (iz.east)--(oc.west);
 \draw[arr,draw=flowred] (1.0,0)--node[right,font=\footnotesize]{$m_1$}(1.5,-1.15);
 \draw[arr,draw=flowred] (2.5,-1.15)--node[right,font=\footnotesize]{$m_2$}(3.0,-2.3);
 \node[anchor=west,font=\footnotesize] at (4.8,0) {$\{x\}$};
 \node[anchor=west,font=\footnotesize] at (4.8,-1.15) {$\{x,y\}$};
 \node[anchor=west,font=\footnotesize] at (4.8,-2.3) {$\{x,y,z\}$};
 \draw[-{Stealth[length=2mm]},flowgray] (0.4,-2.95)--(4.0,-2.95)
   node[midway,below,font=\footnotesize]{time};
\end{tikzpicture}
\label{fig:reach}
\end{figure}
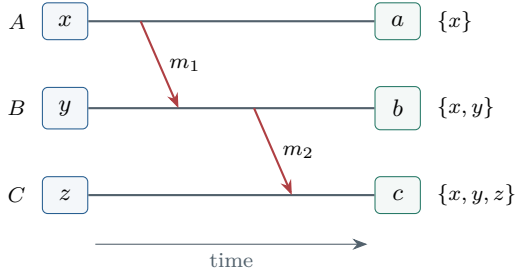

\begin{theorem}
\label{thm:affine}
Let $\mathcal A$ satisfy Definition~\ref{def:arch} and assumptions (i)--(iv). Every behaviour in $\PC_{\mathcal A}$ satisfies
\begin{equation}
 E_S(x)=E_S(x')\quad\text{if}\quad
 x_{R_{\mathcal A}(S)}=x'_{R_{\mathcal A}(S)},
 \label{eq:nodep}
\end{equation}
and these are the only affine constraints. In moment coordinates, the affine hull of $\PC_{\mathcal A}$ is the set of moment vectors in which each $E_S$ is an arbitrary function of $x_{R_{\mathcal A}(S)}$, and
\begin{equation}
 \dim\PC_{\mathcal A}=D_{\mathcal A}(N)
 =\sum_{S\ne\emptyset}\ \prod_{j\in R_{\mathcal A}(S)}d_j .
 \label{eq:archdim}
\end{equation}
Points of this affine hull need not be valid probability distributions.
\end{theorem}

\begin{proof}
Conditional on $\lambda$, each output depends only on its ancestors, so the product over $S$ depends only on $x_{R_{\mathcal A}(S)}$; input-independent mixtures preserve Eq.~\eqref{eq:nodep}.

For the converse, flipping output signs preserves the schedule by assumption~(i). Averaging over all flips gives $\U$, whose moment vector is zero, so the affine hull is the linear span of the attainable moment vectors. If $F_\epsilon$ flips outputs according to $\epsilon\in\{\pm1\}^N$, it multiplies $E_T$ by $\prod_{i\in T}\epsilon_i$, and the signed average
\begin{equation}
 L_S=2^{-N}\sum_\epsilon\Bigl(\prod_{i\in S}\epsilon_i\Bigr)F_\epsilon
 \label{eq:sector}
\end{equation}
keeps the moments $E_S$ and removes all others by character orthogonality, so each $S$ can be treated separately. Fix $S$ and an assignment $r$ of the inputs in $R_{\mathcal A}(S)$. Give input $j$ the flag $h_j=\one\{x_j=r_j\}$ if $j\in R_{\mathcal A}(S)$ and $h_j=1$ otherwise. At each message event the sender transmits the AND of its own flag and all flags received so far, which assumption~(iii) permits. By induction over the schedule, each message is the AND of the flags of its ancestral inputs, and binary symbols suffice even for broadcasts. At output $i\in S$, let $g_i$ be the AND of its own and received flags, so $g_i=\prod_{j\in R_{\mathcal A}(\{i\})}h_j$ and
\begin{equation}
 \prod_{i\in S}g_i=\one\{x_{R_{\mathcal A}(S)}=r\}.
 \label{eq:indicator}
\end{equation}
For each $T\subseteq S$, let the agents in $T$ output the sign $2g_i-1$ and all others $+1$; the resulting $S$-moment is $\prod_{i\in T}(2g_i-1)$, and
\begin{equation}
 \prod_{i\in S}g_i=2^{-|S|}\sum_{T\subseteq S}\prod_{i\in T}(2g_i-1)
 \label{eq:flagspan}
\end{equation}
places each indicator~\eqref{eq:indicator} in the span, from which $L_S$ isolates it. These indicators span all functions of $x_{R_{\mathcal A}(S)}$, and summing over $S$ gives Eq.~\eqref{eq:archdim}.
\end{proof}
\begin{remark}
\label{rem:onebit}
One might expect a narrow channel to narrow the affine hull, since a single binary message cannot transmit an input alphabet with more than two symbols. Theorem~\ref{thm:affine} shows otherwise: one binary symbol per message event already gives each statistic its full reach. The spanning argument uses many protocols, each transmitting only the one-bit flag for a single assignment $r$, and the resulting indicators span all functions of $x_{R_{\mathcal A}(S)}$ although no single protocol carries the whole input. Capacity therefore restricts the polytope but not its affine hull, and schedules with equal reachability but different alphabets can simulate different behaviours. The ternary example below and the two-bit relay of Sec.~\ref{sec:three} are instances.
\end{remark}
\begin{corollary}
\label{cor:eqbasis}
Fix $S$, write $R=R_{\mathcal A}(S)$ and $x=(x_R,x_{R^c})$, and choose a reference value $x^\circ_{R^c}$ of the unreachable inputs. The equalities
\begin{equation}
 E_S(x_R,x_{R^c})=E_S(x_R,x^\circ_{R^c}),\quad x_{R^c}\ne x^\circ_{R^c},
\label{eq:eqbasis}
\end{equation}
over all $S$ and $x_R$, are independent and complete, and there are $D_{\mathrm{full}}-D_{\mathcal A}(N)$ of them. Keeping only moments of order at most $K$, the image of $\PC_{\mathcal A}$ has dimension $D_{\mathcal A}(K)$.
\end{corollary}
\begin{proof}
Each equality involves a distinct non-reference coordinate, so the equalities are independent; their number equals the codimension in Theorem~\ref{thm:affine}, so they are complete. Projection onto the moments with $|S|\le K$ maps the affine hull onto the sum of the retained function spaces.
\end{proof}

For the local schedule $R(S)=S$, and Eq.~\eqref{eq:archdim} reproduces $D_0(N)$. If every input reaches every output, the dimension is $D_{\mathrm{full}}$, which does not imply that every behaviour is achievable. In the example above, $R(\{A\})=\{x\}$ and $R(\{B\})=R(\{A,B\})=\{x,y\}$, giving $2+4+4=10$, and Eq.~\eqref{eq:eqbasis} with reference $y=0$ is Eq.~\eqref{eq:example}.

Let $x,y\in\{0,1,2\}$ and allow one message from Alice to Bob. Binary and ternary
message alphabets both give dimension $3+9+9=21$. Nevertheless, the target
\begin{equation}
 a=0,\qquad b=\one\{x=y\}
 \label{eq:capacity}
\end{equation}
requires three message symbols for exact simulation. With two symbols, two values of $x$ share a message, so Bob's response rows for these values coincide, whereas the target rows differ; a mixture that reproduces a deterministic behaviour must reproduce it in every positive-weight component, so shared randomness does not help. This is a one-way instance of the equality function, a standard example in communication complexity~\cite{KN1997}.

\subsection{Convex certificates}
\label{sec:cert}

Choose a coordinate map and let $u_\lambda$ be the image of each distinct
deterministic behaviour of a declared schedule. A target with moment vector
$\bm m$ belongs to the class exactly when
\begin{equation}
 \begin{aligned}
 \bm m&=\sum_\lambda\mu_\lambda u_\lambda,\\
 \sum_\lambda\mu_\lambda&=1,\qquad\mu_\lambda\ge0.
 \end{aligned}
 \label{eq:membership}
\end{equation}
The weights give a simulation in which shared randomness selects protocol
$\lambda$ with probability $\mu_\lambda$; in projected coordinates they certify
only the retained statistics. Otherwise, linear-programming duality gives a
separating inequality~\cite{Brunner2014},
\begin{equation}
 \bm c^\top u_\lambda\le\beta\quad\forall\lambda,\qquad
 \bm c^\top\bm m>\beta.
 \label{eq:separation}
\end{equation}
An equality of Corollary~\ref{cor:eqbasis}, written $\bm w^\top u_\lambda=0$ for
all $\lambda$, that the target violates is a separator with $\beta=0$ after
choosing the sign of $\bm w$; its contexts identify a dependence on an
unreachable input. Figure~\ref{fig:flow} summarizes the procedure.

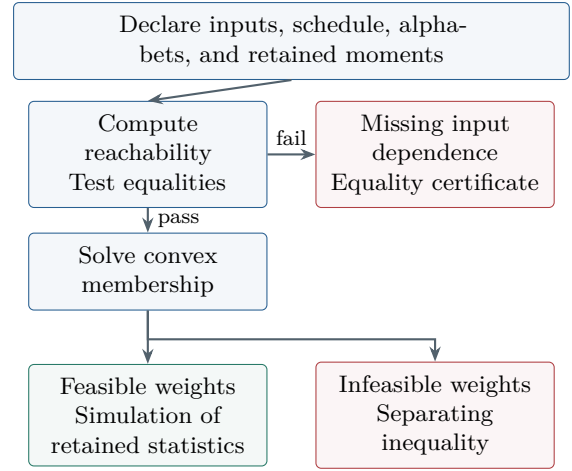
\begin{figure}[tb]
\centering
\caption{Certification procedure. Passing the affine equalities permits a convex membership test but does not imply simulation. A decomposition in complete coordinates certifies the full behaviour; after projection it certifies only the retained moments.}
\begin{tikzpicture}
 \node[box,text width=7.0cm] (start) at (0,0)
 {Declare inputs, schedule, alphabets, and retained moments};
 \node[box,text width=2.8cm] (check) at (-1.9,-1.5)
 {Compute reachability\\Test equalities};
 \node[box,text width=2.8cm,draw=flowred,fill=flowred!5] (fail) at (1.9,-1.5)
 {Missing input dependence\\Equality certificate};
 \node[box,text width=2.8cm] (lp) at (-1.9,-3.05)
 {Solve convex membership};
 \node[box,text width=2.8cm,draw=flowgreen,fill=flowgreen!5] (yes) at (-1.9,-4.95)
 {Feasible weights\\Simulation of retained statistics};
 \node[box,text width=2.8cm,draw=flowred,fill=flowred!5] (no) at (1.9,-4.95)
 {Infeasible weights\\Separating inequality};
 \draw[arr] (start.south)--(check.north);
 \draw[arr] (check)--node[above,font=\footnotesize] {fail}(fail);
 \draw[arr] (check)--node[right,font=\footnotesize] {pass}(lp);
 \draw[arr] (lp)--(yes);
 \draw[arr] (lp.south)--++(0,-0.38)-|(no.north);
\end{tikzpicture}
\label{fig:flow}
\end{figure}

For a noisy interpolation $\bm m_\eta=(1-\eta)\bm m+\eta\,\bm m_{\mathrm{ref}}$ with $\bm c^\top\bm m_{\mathrm{ref}}\le\beta$, the witness remains violated exactly for $0\le\eta<\eta^\star$, where
\begin{equation}
 \eta^\star=\frac{\bm c^\top\bm m-\beta}
 {\bm c^\top\bm m-\bm c^\top\bm m_{\mathrm{ref}}}.
 \label{eq:noise}
\end{equation}
This is the robustness of that certificate, not necessarily the exact membership
threshold; an equality violation diluted toward an equality-compatible reference
remains nonzero for every $\eta<1$.

For estimated moments $\widehat{\bm m}$ with estimated covariance
$\widehat\Sigma$, a standardized margin is
\begin{equation}
 \zeta=\frac{\bm c^\top\widehat{\bm m}-\beta}
 {\sqrt{\bm c^\top\widehat\Sigma\bm c}},\qquad
 \bm c^\top\widehat\Sigma\bm c>0 .
 \label{eq:z}
\end{equation}
It is defined only when the estimated variance is positive and is a signal-to-noise ratio rather than a calibrated significance level. For a mean $\widehat E$ of $n$ independent signed products at a fixed context, Hoeffding's inequality~\cite{Hoeffding1963} gives $\Pr\{|\widehat E-E|\ge\varepsilon\}\le2\exp(-n\varepsilon^2/2)$, so for two contexts with equal true moments a union bound gives, with probability at least $1-\alpha$,
\begin{equation}
 |\widehat E-\widehat E'|
 \le\sqrt{\frac{2\log(4/\alpha)}n}
 +\sqrt{\frac{2\log(4/\alpha)}{n'}}.
 \label{eq:eqconfidence}
\end{equation}
Testing many equalities, or choosing a witness on the same data, requires multiplicity control or independent validation data, and serially dependent trials require a different analysis. Positive rescaling of $(\bm c,\beta)$ leaves $\zeta$ and $\eta^\star$ unchanged. Minimizing $\|\bm c\|_1$ at fixed margin is one way to obtain sparser witnesses, although it does not guarantee fewer settings or a larger $\zeta$.

\section{Communication hierarchies}
\label{sec:comm}

All inputs and outputs in this section are binary. A ``round'' is one prescribed message event, not a simultaneous distributed round, and every level can ignore any message, so the levels are nested.

\subsection{Canonical targets}
\label{sec:targets}

For two agents, define
\begin{equation}
 \mathcal B=E_{AB}(00)+E_{AB}(01)+E_{AB}(10)-E_{AB}(11).
 \label{eq:chsh}
\end{equation}
Local behaviours satisfy $|\mathcal B|\le2$~\cite{CHSH1969}, and in this scenario the eight relabelled CHSH inequalities together with positivity characterize the local polytope~\cite{Fine1982,Brunner2014}. The PR behaviour $P_{\mathrm{PR}}(ab\mid xy)=\tfrac12\one\{a\oplus b=xy\}$~\cite{PR1994} has zero marginals, $E_{AB}(xy)=(-1)^{xy}$, and $\mathcal B=4$, above the quantum bound $2\sqrt2$~\cite{Cirelson1980}. The family $P_{\mathrm{PR},v}=vP_{\mathrm{PR}}+(1-v)\U$ has correlators $v(-1)^{xy}$ and largest CHSH value $4v$, so it is local exactly for $v\le1/2$. It is quantum-realizable exactly for $v\le1/\sqrt2$, being the mixture $\sqrt2v\,P_{\mathrm{Ts}}+(1-\sqrt2v)\U$ for such $v$, where the Tsirelson target $P_{\mathrm{Ts}}$ has zero marginals and $E_{AB}(xy)=(-1)^{xy}/\sqrt2$. The latter is realized on $|\Phi^+\rangle=(|00\rangle+|11\rangle)/\sqrt2$ with Alice
measuring $\sigma_z$ or $\sigma_x$ and Bob measuring
$(\sigma_z\pm\sigma_x)/\sqrt2$.

Figure~\ref{fig:chsh} shows the projection onto $\mathcal B$ and $\mathcal B'=E_{AB}(00)+E_{AB}(01)-E_{AB}(10)+E_{AB}(11)$. The projections of the local, quantum, and no-signalling sets are
\begin{align}
 |\mathcal B|,|\mathcal B'|&\le2, \nonumber\\
 \mathcal B^2+\mathcal B'^2&\le8,\nonumber\\
 |\mathcal B\pm\mathcal B'|&\le4.
 \label{eq:images}
\end{align}
At a local vertex with signed tables $A(x)$, $B(y)$, either $B(0)=B(1)$ and $\mathcal B=\mathcal B'=2A(0)B(0)$, or $B(0)=-B(1)$ and $\mathcal B=-\mathcal B'=2A(1)B(0)$, so the local image is the square with corners $(\pm2,\pm2)$. With $U=E_{00}+E_{01}$ and $V=E_{10}-E_{11}$, the no-signalling image is $|U|,|V|\le2$, attained by zero-marginal correlators. The unit-vector representation $E_{xy}=\bm s_x\cdot\bm t_y$ of quantum correlators~\cite{Cirelson1980} gives $U^2+V^2\le\|\bm t_0+\bm t_1\|^2+\|\bm t_0-\bm t_1\|^2=4$, with equality attained by coplanar Pauli measurements on a maximally entangled pair; convex mixing fills the disk. Membership in a projected set does not establish membership of the unprojected behaviour.

\begin{figure}[tb]
\centering
\caption{Images in the $(\mathcal B,\mathcal B')$ plane: the local square, quantum disk, and no-signalling diamond. The dashed line is the witness $\mathcal B=2$. The PR, Tsirelson, and noisy-PR points lie on $\mathcal B'=0$.}
\begin{tikzpicture}[x=0.72cm,y=0.72cm,font=\small]
 \fill[flowgray!9] (4,0)--(0,4)--(-4,0)--(0,-4)--cycle;
 \fill[flowblue!14] (0,0) circle[radius=2.828427];
 \fill[flowgreen!18] (-2,-2) rectangle (2,2);
 \draw[flowgray,thick] (4,0)--(0,4)--(-4,0)--(0,-4)--cycle;
 \draw[flowblue,thick] (0,0) circle[radius=2.828427];
 \draw[flowgreen,thick] (-2,-2) rectangle (2,2);
 \draw[->] (-4.45,0)--(4.6,0) node[right] {$\mathcal B$};
 \draw[->] (0,-4.3)--(0,4.5) node[above] {$\mathcal B'$};
 \draw[flowred,dashed,thick] (2,-3.9)--(2,3.9);
 \node[flowred,anchor=west] at (2.06,3.7) {$\mathcal B=2$};
 \foreach \i in {-4,-2,2,4}{\draw(\i,0.07)--(\i,-0.07);
  \node[below,font=\footnotesize] at (\i,-0.05){$\i$};}
 \foreach \i in {-4,-2,2,4}{\draw(0.07,\i)--(-0.07,\i);}
 \fill[flowred] (4,0) circle[radius=0.075];
 \node[flowred,anchor=south east] at (4.8,0.1) {PR};
 \fill[flowblue] (2.828427,0) circle[radius=0.075];
 \draw[flowblue] (2.828427,0.12)--(3.55,1.25);
 \node[flowblue,align=center,anchor=south] at (3.5,1.30)
 {Tsirelson\\$(2\sqrt2,0)$};
 \fill[flowgreen] (2,0) circle[radius=0.075];
 \node[flowgreen,anchor=north west] at (0.05,-0.52) {$v=1/2$};
 \node[flowgreen] at (-1.1,-1.45) {Local};
 \node[flowblue] at (-2.95,2.27) {Quantum};
 \node[flowgray] at (-1.75,3.9) {No-signalling};
\end{tikzpicture}
\label{fig:chsh}
\end{figure}
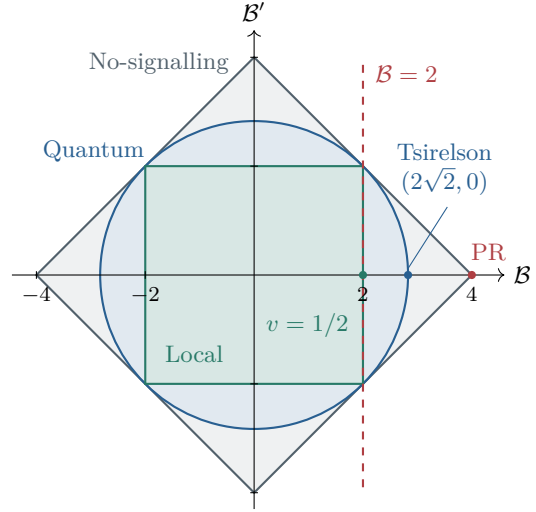

For three agents, Mermin's expression is
\begin{equation}
 \begin{aligned}
 \mathcal M={}&E_{ABC}(000)-E_{ABC}(011)\\
 &-E_{ABC}(101)-E_{ABC}(110),
 \end{aligned}
 \label{eq:mermin}
\end{equation}
and local behaviours satisfy $|\mathcal M|\le2$~\cite{Mermin1990}. The state $|\mathrm{GHZ}_3\rangle=(|000\rangle+|111\rangle)/\sqrt2$ \cite{GHSZ1990}, with input $0$ denoting $\sigma_x$ and input $1$ denoting $\sigma_y$, has vanishing one- and two-body moments and
\begin{equation}
 E_{ABC}(x,y,z)=\cos\bigl[\tfrac\pi2(x+y+z)\bigr],
 \label{eq:ghz3}
\end{equation}
so $\mathcal M=4$. The family $P_{\mathrm{GHZ},v}=vP_{\mathrm{GHZ}}+(1-v)\U$ is local exactly for $v\le1/2$. Necessity follows from $\mathcal M(P_{\mathrm{GHZ},v})=4v$. For sufficiency, the weights $\mu_\lambda=[1+\mathcal M(\lambda)/2]/64$ on the $64$ local vertices are nonnegative by Mermin's bound and sum to one because $\mathcal M$ averages to zero; since distinct monomials in the signed responses are orthogonal under the uniform distribution on vertices, the mixture has half of each nonzero GHZ moment and zero for every other moment.

\subsection{Two agents}
\label{sec:two}

Let $\PC^{(2)}_k$ be the two-agent class with Alice speaking first:
\begin{align}
 k=0:&\quad a=a(x),\quad b=b(y),\nonumber\\
 k=1:&\quad m_1=m_1(x),\quad a=a(x),\quad b=b(y,m_1),\nonumber\\
 k=2:&\quad m_2=m_2(y,m_1),\nonumber\\
     &\quad a=a(x,m_2),\quad b=b(y,m_1).
 \label{eq:two}
\end{align}
We retain $E_A$, $E_B$, and $E_{AB}$ at all four contexts, giving twelve
coordinates without a no-signalling assumption.

\begin{proposition}
\label{prop:two}
The hierarchy is strict, $\PC^{(2)}_0\subsetneq\PC^{(2)}_1\subsetneq\PC^{(2)}_2$, with $16$, $64$, and $256$ vertices and affine dimensions $8$, $10$, and $12$. The one-message class consists exactly of the behaviours whose Alice marginal is independent of $y$; the two-message class is the full normalized behaviour polytope.
\end{proposition}
\begin{proof}
With $m_1=x$, $b$ is an arbitrary function of $(x,y)$, and with $m_1=x$, $m_2=y$ both outputs are; this gives the vertex counts, the $k=1$ count agreeing with the one-bit vertex formula of Ref.~\cite{Cruzeiro2019}. The dimensions follow from Theorem~\ref{thm:affine} with reach sets $(\{x\},\{y\},\{x,y\})$, $(\{x\},\{x,y\},\{x,y\})$, and $(\{x,y\},\{x,y\},\{x,y\})$ for $(E_A,E_B,E_{AB})$, and give strictness. For the one-message class, necessity is immediate because Alice's output depends only on $x$ and $\lambda$. For sufficiency, let $P_A(a\mid x)=P(a\mid x,y)$ and $P_{B\mid A}(b\mid a,x,y)=P(a,b\mid x,y)/P_A(a\mid x)$ whenever $P_A(a\mid x)>0$. Let $\lambda$ consist of independent components $\alpha_x\sim P_A(\cdot\mid x)$ and $\beta_{x,y,a}\sim P_{B\mid A}(\cdot\mid a,x,y)$; Alice outputs $\alpha_x$ and sends $m_1=x$, and Bob, who knows $x$ and $\lambda$, outputs $\beta_{x,y,\alpha_x}$. Each $\lambda$ fixes a deterministic $k=1$ protocol, and
\begin{equation}
 \Pr(a,b\mid x,y)=P_A(a\mid x)\,P_{B\mid A}(b\mid a,x,y)=P(a,b\mid x,y).
 \label{eq:condsampling}
\end{equation}
When $P_A(a\mid x)=0$ any distribution may be assigned to $\beta_{x,y,a}$, since that branch has probability zero and $P(a,b\mid x,y)=0$ for every $b$. In the two-message class both agents know $(x,y)$, and sampling $(a,b)$ from $P(\cdot,\cdot\mid x,y)$ with shared randomness yields every behaviour.
\end{proof}

The finite-input observation agrees with earlier communication-assisted treatments~\cite{BaconToner2003,Cruzeiro2019}. It differs from Toner and Bacon's one-bit simulation of arbitrary projective measurements on a singlet~\cite{TonerBacon2003}: a binary input is sent here, and no claim is made that one bit transmits an arbitrary finite input alphabet. Figure~\ref{fig:two} shows the information flow.

\begin{figure}[tb]
\centering
\caption{Two-agent hierarchy with Alice speaking first. Choosing the messages to be the speakers' inputs realizes all deterministic tables with the displayed dependencies; the affine dimensions are $8$, $10$, and $12$.}
\begin{tikzpicture}[font=\small]
 \foreach \j/\k in {0/0,1/1,2/2}{
  \node[anchor=east,font=\bfseries] at (-3.15,-1.55*\j) {$k=\k$};
  \node[box,text width=1.85cm] (A\j) at (-1.65,-1.55*\j) {$A$ receives $x$};
  \node[box,text width=1.85cm] (B\j) at (1.95,-1.55*\j) {$B$ receives $y$};
 }
 \node[below=2pt of A0] {$a(x)$};
 \node[below=2pt of B0] {$b(y)$};
 \draw[arr] (A1)--node[above,font=\footnotesize] {$m_1=x$}(B1);
 \node[below=2pt of A1] {$a(x)$};
 \node[below=2pt of B1] {$b(x,y)$};
 \draw[arr] (A2.north east) to[bend left=15]
 node[above,font=\footnotesize] {$m_1=x$}(B2.north west);
 \draw[arr] (B2.south west) to[bend left=15]
 node[below,font=\footnotesize] {$m_2=y$}(A2.south east);
 \node[below=17pt of A2] {$a(x,y)$};
 \node[below=17pt of B2] {$b(x,y)$};
\end{tikzpicture}
\label{fig:two}
\end{figure}
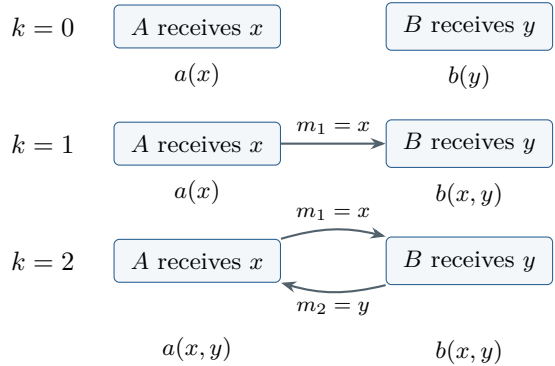

For PR, a uniform shared bit $s$ gives the one-message protocol
\begin{equation}
 a=s,\qquad m_1=x,\qquad b=s\oplus xy.
 \label{eq:prprotocol}
\end{equation}
The targets $a=0,b=x$ and $a=y,b=0$ diagnose the direction of information flow, with equality violations
\begin{equation}
 E_B(0,0)-E_B(1,0)=2,\qquad E_A(0,0)-E_A(0,1)=2 .
 \label{eq:twocert}
\end{equation}
The first difference vanishes on $\PC^{(2)}_0$ and the second on both $\PC^{(2)}_0$ and $\PC^{(2)}_1$ (Table~\ref{tab:two}). The last target needs two events only because of the prescribed order, so the minimum within a fixed hierarchy is not an intrinsic cost of the target over all protocols.

\begin{table}[tb]
\centering
\caption{Two-agent target membership. $k_{\min}$ counts events in
Eq.~\eqref{eq:two}. Noisy PR rows use
$P_{\mathrm{PR},v}=vP_{\mathrm{PR}}+(1-v)\U$.}
\label{tab:two}
\small
\begin{tabular}{lcccc}
\toprule
Target & $\PC^{(2)}_0$ & $\PC^{(2)}_1$ & $\PC^{(2)}_2$ & $k_{\min}$\\
\midrule
Uniform & Y & Y & Y & 0\\
Noisy PR, $v=1/2$ & Y & Y & Y & 0\\
Noisy PR, $v=7/10$ & N & Y & Y & 1\\
Tsirelson & N & Y & Y & 1\\
PR & N & Y & Y & 1\\
$a=0,\ b=x$ & N & Y & Y & 1\\
$a=y,\ b=0$ & N & N & Y & 2\\
\bottomrule
\end{tabular}
\end{table}

\subsection{Three agents: broadcast and relay}
\label{sec:three}

In the broadcast hierarchy $\PC^{\mathrm{bc}}_k$, one bit is one symbol delivered to both other agents; Alice broadcasts first and Bob second:
\begin{align}
 k=0:&\quad a(x),\ b(y),\ c(z),\nonumber\\
 k=1:&\quad m_1(x);\quad a(x),\ b(y,m_1),\ c(z,m_1),\nonumber\\
 k=2:&\quad m_2(y,m_1);\quad a(x,m_2),\ b(y,m_1),\nonumber\\
     &\hspace{39mm}c(z,m_1,m_2).
 \label{eq:broadcast}
\end{align}
In the point-to-point relay hierarchy $\PC^{\mathrm{pt}}_k$, the events are Alice to Bob and then Bob to Charlie (Fig.~\ref{fig:reach}):
\begin{align}
 k=0:&\quad a(x),\ b(y),\ c(z),\nonumber\\
 k=1:&\quad m_1(x);\quad a(x),\ b(y,m_1),\ c(z),\nonumber\\
 k=2:&\quad m_2(y,m_1);\nonumber\\
     &\quad a(x),\ b(y,m_1),\ c(z,m_2).
 \label{eq:relay}
\end{align}
Charlie is silent in both families, and all seven moments at all eight contexts are retained.

\begin{table*}[tb]
\centering
\caption{Reach sets, affine dimensions, and numbers of distinct deterministic behaviours for three-agent schedules. A string such as $xy$ denotes the input set $\{x,y\}$.}
\label{tab:reach}
\begin{tabular}{lcccccccrr}
\toprule
Schedule & $A$ & $B$ & $C$ & $AB$ & $AC$ & $BC$ & $ABC$ & Dimension & Vertices\\
\midrule
No messages & $x$ & $y$ & $z$ & $xy$ & $xz$ & $yz$ & $xyz$ & 26 & 64\\
One broadcast & $x$ & $xy$ & $xz$ & $xy$ & $xz$ & $xyz$ & $xyz$ & 34 & 1024\\
Two broadcasts & $xy$ & $xy$ & $xyz$ & $xy$ & $xyz$ & $xyz$ & $xyz$ & 44 & 65536\\
One relay bit & $x$ & $xy$ & $z$ & $xy$ & $xz$ & $xyz$ & $xyz$ & 32 & 256\\
Two relay bits & $x$ & $xy$ & $xyz$ & $xy$ & $xyz$ & $xyz$ & $xyz$ & 42 & 5632\\
\bottomrule
\end{tabular}
\end{table*}

\begin{proposition}
\label{prop:three}
Both three-agent hierarchies are strictly nested at levels $0,1,2$, with the dimensions and vertex counts of Table~\ref{tab:reach}, and $\PC^{\mathrm{pt}}_k\subsetneq\PC^{\mathrm{bc}}_k$ for $k=1,2$.
\end{proposition}
\begin{proof}
The dimensions follow from Theorem~\ref{thm:affine}. For broadcast, $m_1=x$ and $m_2=y$ give arbitrary tables with the indicated dependencies. For two relay bits, $a(x)$ and $b(x,y)$ are arbitrary, but $c(x,y,z)$, viewed as four rows indexed by $(x,y)$, has at most two distinct rows, giving $4+\binom42(2^4-2)=88$ Charlie tables and $4\cdot16\cdot88=5632$ behaviours. Ignoring a received copy of a broadcast emulates the relay, and the unequal dimensions give strictness.
\end{proof}

The relay illustrates the difference between affine dimension and capacity: its $42$ dimensions are those of a schedule in which Charlie receives the full label $(x,y)$, yet one bit permits only $88$ of $256$ Charlie tables. The table whose rows are $00,01,10,11$ for $(x,y)=00,01,10,11$ satisfies every relay equality but is not achievable.

\paragraph{One-bit GHZ simulation.}
The $\sigma_x/\sigma_y$ GHZ target belongs to both one-message classes. That one bit suffices follows from Tessier et al., who reproduce all Pauli-product correlations of the $n$-qubit GHZ state with $n-2$ bits, optimal for that measurement set~\cite{Tessier2005}. We give a realization within the declared hierarchies. Let $s_0,s_1,s_2$ be independent unbiased shared signs and
\begin{equation}
 \begin{array}{c|cc}
 G_{xy} & y=0 & y=1\\\hline
 x=0 & s_1 & -s_2\\
 x=1 & -s_2 & -s_1
 \end{array}
 \qquad C(0)=s_1,\quad C(1)=s_2.
 \label{eq:ghzprotocol}
\end{equation}
Alice outputs $s_0$ and sends $x$ to Bob, Bob outputs $s_0G_{xy}$, and Charlie outputs $C(z)$. Every one- and two-body moment contains an unpaired unbiased sign, and $\mathbb E(G_{xy}C(z))$ equals $+1$ at $000$, $-1$ at $011$, $101$, and $110$, and $0$ elsewhere, which is Eq.~\eqref{eq:ghz3}.

\paragraph{Further targets.}
For PR on $A,B$ with $c=0$, Eq.~\eqref{eq:prprotocol} applies to either one-message schedule. For $a=b=0$, $c=x\oplus y$,
\begin{equation}
 E_C(0,0,0)-E_C(0,1,0)=2,
 \label{eq:ccert}
\end{equation}
which vanishes in both one-message classes; two broadcasts send $x$ and $y$, and the two-bit relay sends $x$ and then $x\oplus y$. Broadcast and relay differ for $a=b=0$, $c=x$, where one broadcast suffices but two relay bits are needed, and for $a=y$, $b=c=0$, which the prescribed relay never achieves (Table~\ref{tab:three}). For $a=z$, $b=c=0$,
\begin{equation}
 E_A(0,0,0)-E_A(0,0,1)=2
 \label{eq:zcert}
\end{equation}
vanishes in every schedule in which Charlie is silent. The certificate identifies the required dependence, information about $z$ reaching Alice before her output, but not a unique physical channel. If Charlie also broadcasts $z$ after the two broadcasts of Eq.~\eqref{eq:broadcast} and all outputs follow this event, every deterministic table becomes achievable and the class is the full product of eight seven-dimensional simplices.

\begin{table*}[tb]
\centering
\caption{Three-agent membership and scheduled minimum event counts. A dash means that no level $k\le2$ realizes the target. $P_{\mathrm{GHZ},v}=vP_{\mathrm{GHZ}}+(1-v)\U$.} \label{tab:three}
\begin{tabular}{lccccccc}
\toprule
Target & $\PC_0$ & $\PC^{\mathrm{bc}}_1$ & $\PC^{\mathrm{bc}}_2$ & $k_{\min}^{\mathrm{bc}}$
 & $\PC^{\mathrm{pt}}_1$ & $\PC^{\mathrm{pt}}_2$ & $k_{\min}^{\mathrm{pt}}$\\
\midrule
Uniform & Y & Y & Y & 0 & Y & Y & 0\\ Noisy GHZ, $v=1/2$ & Y & Y & Y & 0 & Y & Y & 0\\
GHZ, $\sigma_x/\sigma_y$ & N & Y & Y & 1 & Y & Y & 1\\
PR on $A,B$; $c=0$ & N & Y & Y & 1 & Y & Y & 1\\
$a=b=0,\ c=x$ & N & Y & Y & 1 & N & Y & 2\\
$a=b=0,\ c=x\oplus y$ & N & N & Y & 2 & N & Y & 2\\
$a=y,\ b=c=0$ & N & N & Y & 2 & N & N & --\\
$a=z,\ b=c=0$ & N & N & N & -- & N & N & --\\
\bottomrule
\end{tabular}
\end{table*}

\section{Low-order projections}
\label{sec:loss}

For $N\ge2$, consider $|\mathrm{GHZ}_N\rangle=(|0\rangle^{\otimes N}+|1\rangle^{\otimes N})/\sqrt2$ with each agent measuring $\sigma_x$ (input $0$) or $\sigma_y$ (input $1$). Every proper reduced state equals the corresponding reduction of the separable mixture $\rho_{\mathrm{mix}}=\tfrac12(|0\rangle\langle0|^{\otimes N}+|1\rangle\langle1|^{\otimes N})$, so correlations among fewer than $N$ parties do not reveal the GHZ coherence. This obstruction, for every order $k<N$ and without a symmetry restriction, is given in Ref.~\cite[Sec.~VI]{Tan2016}; the related Bell-inequality observation appears in Ref.~\cite{Tura2015}. Since $\rho_{\mathrm{mix}}$ is separable, the order-$K$ projection of the GHZ behaviour lies in $\CC_K$ for every $K<N$ and every choice of settings.

For $\sigma_x/\sigma_y$ settings the moments are explicit. The reductions of $\rho_{\mathrm{mix}}$ are diagonal in the computational basis, whereas nonempty products of $\sigma_x$ and $\sigma_y$ are off-diagonal, so all proper moments vanish. The GHZ state is an eigenvector of each $N$-fold product with an even number $n_y$ of $\sigma_y$ factors~\cite{Mermin1990}, with eigenvalue $(-1)^{n_y/2}$, and products with odd $n_y$ map it to an orthogonal state. Hence
\begin{align}
 E_S(x_S)&=0\quad(\emptyset\ne S\subsetneq[N]),\nonumber\\
 E_{[N]}(x)&=\cos\Bigl(\tfrac\pi2\sum_i x_i\Bigr).
 \label{eq:ghzall}
\end{align}

\begin{proposition}
\label{prop:collapse}
For every $K<N$, $\Pi_K(P_{\mathrm{GHZ}})=\Pi_K(\U)=0$. This point is the centre of the $\ell_1$ ball of Eq.~\eqref{eq:ball}, so every valid inequality $\bm c^\top u\le\beta$ for $\CC_K$ with $\bm c\ne0$ has slack at least $\|\bm c\|_\infty$ there.
\end{proposition}
\begin{proof}
The retained GHZ moments vanish by Eq.~\eqref{eq:ghzall}, the moments of $\U$ vanish, and Lemma~\ref{lem:l1} gives the remaining statements.
\end{proof}

The qualitative content of Proposition~\ref{prop:collapse} is the obstruction of Ref.~\cite{Tan2016}. The proposition adds an explicit moment-space proof for $\sigma_x/\sigma_y$ settings and a quantitative margin: every valid inequality has slack at least $\|\bm c\|_\infty$ at the common image, and a Euclidean ball of radius $1/\sqrt{D_0(K)}$ lies inside $\CC_K$. At $N=3$, the target violates Mermin's inequality maximally but lies at the origin of the $18$-dimensional second-order projection (Fig.~\ref{fig:collapse}), with guaranteed Euclidean radius $1/\sqrt{18}$; this radius is a lower bound, not the inradius.

\begin{figure}[tb]
\centering
\caption{Distinct complete behaviours with identical one- and two-body moments. The ellipse schematically represents a neighbourhood of the origin in an $18$-dimensional polytope; the three-body moments distinguish the two sources.}
\begin{tikzpicture}[font=\small]
 \node[box,text width=2.85cm] (ghz) at (-2.0,1.15)
 {GHZ, $\sigma_x/\sigma_y$ settings\\$\mathcal M=4$\\Proper moments $=0$};
 \node[box,text width=2.85cm,draw=flowgreen,fill=flowgreen!5] (uni) at (-2.0,-1.15)
 {Uniform local behaviour\\$\mathcal M=0$\\All moments $=0$};
 \draw[flowgreen,fill=flowgreen!7,thick] (2.0,0) ellipse[x radius=1.6cm,y radius=1.1cm];
 \fill[flowgreen] (2.0,0) circle[radius=0.055cm];
 \node[above right] at (2.03,0.05) {$u=0$};
 \node[align=center,text width=3.2cm] at (2.0,-1.7)
 {Interior of $\Pi_2(\PC_0)$\\$D_0(2)=18$};
 \draw[arr] (ghz.east)--(1.9,0.12);
 \draw[arr] (uni.east)--(1.9,-0.12);
 \node[fill=white,inner sep=2pt] at (0.25,0.78) {$\Pi_2$};
 \node[fill=white,inner sep=2pt] at (0.25,-0.78) {$\Pi_2$};
 \node[font=\footnotesize,flowgray] at (2.0,1.43) {Schematic projection};
\end{tikzpicture}
\label{fig:collapse}
\end{figure}
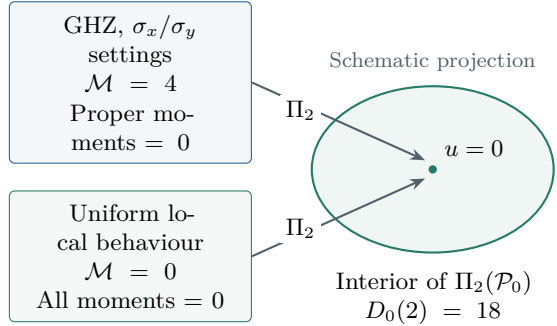

\begin{definition}
A no-signalling target $P$ is $K$-body detectable if some valid inequality for the local model supported on moments of order at most $K$ is violated by $P$, equivalently if $\Pi_K(P)\notin\CC_K$.
\end{definition}

Detectability is decided by the membership program~\eqref{eq:membership}, and detectability at order $K$ implies detectability at every larger order. The three-agent $\sigma_x/\sigma_y$ GHZ target requires $K=3$, and at $N=2$ the $\sigma_x/\sigma_y$ GHZ behaviour is itself local. The quantitative statement depends on the settings: measurements involving $\sigma_z$ can give nonzero proper moments, which still agree with those of $\rho_{\mathrm{mix}}$ but need not lie at an interior point. Low-order Bell tests do detect other entangled states~\cite{Tura2014}. A classical analogue shows the same loss: if independent unbiased bits $a,b$ determine $c=a\oplus b$, every pair is independent while $a\oplus b\oplus c=0$ exactly.

\section{Agents and observers}
\label{sec:agents}

\subsection{Local decision rules}
\label{sec:qdt}

In decision terms, $x_i$ is the information available to agent $i$ and $a_i$ its choice; a behaviour is local if it can be explained by information shared before the task, without exchange during it. Consider a binary response rule motivated by quantum decision theory,
\begin{equation}
 p_i(1\mid x_i,\lambda)=f_i(x_i,\lambda)+\kappa\,q_i(x_i,\lambda)\in[0,1],
 \label{eq:qdtrule}
\end{equation}
where the utility and attraction decomposition follows the decision-theory literature~\cite{YS2010,Favre2016,Yukalov2020}, and $\kappa$ is an interpolation parameter introduced here, not a parameter claimed for those models. A value is admissible if every response probability lies in $[0,1]$.

\begin{corollary}
\label{cor:qdt}
If each response uses only its own input and input-independent shared randomness, and residual random choices are conditionally independent given that randomness, the joint behaviour lies in $\PC_0$ for every admissible $\kappa$, and its communication cost in the present resource model is zero.
\end{corollary}

Under these assumptions an interference term can move the behaviour within the local polytope but cannot violate a Bell inequality. Yukalov and Sornette describe prospect operators as entangling because of the composite structure of a decision maker's action prospects~\cite{YS2010}; this does not by itself identify the Hilbert-space factors with separate agents holding independently chosen inputs, and an entangled representation alone does not imply Bell nonlocality~\cite{Brunner2014}. A distributed model that specifies a Bell-nonlocal joint law must account for a shared quantum resource, joint access to private inputs, or information exchange; network models with repeated exchange~\cite{Yukalov2018} correspond to a different schedule. The same applies to the classical networks of Refs.~\cite{YY2026a,YY2026b}: there the initial single-step choices are made before any exchange, whereas each later step updates a group's choice probability using the other groups' probabilities, which in the present terms is a message event. Those models describe aggregate choice fractions without setting choices, so the behaviour-level certificates of Sec.~\ref{sec:cert} would apply only after inputs are introduced.

A complementary construction corrects the joint distribution rather than the individual rules. For two binary responses, initially independent with common marginal $(\nu,1-\nu)$, $0<\nu<1$, add $\delta$ to the diagonal and subtract it from the off-diagonal entries:
\begin{equation}
 \begin{aligned}
 P_\delta=\bigl(&\nu^2+\delta,\ \nu(1-\nu)-\delta,\\
 &\nu(1-\nu)-\delta,\ (1-\nu)^2+\delta\bigr),
 \end{aligned}
 \label{eq:pdelta}
\end{equation}
in the order $00,01,10,11$. Both marginals remain $(\nu,1-\nu)$, and for these marginals Eq.~\eqref{eq:pdelta} is the general joint distribution, with nonnegativity requiring $-\min\{\nu^2,(1-\nu)^2\}\le\delta\le\nu(1-\nu)$. In signed moments, $E_A=E_B=2\nu-1$ and $E_{AB}=(2\nu-1)^2+4\delta$, so only the connected two-body moment changes, and the Pearson correlation coefficient is $\delta/[\nu(1-\nu)]$. At $\delta=\nu(1-\nu)$ the outputs agree on every trial and their mutual information equals the binary entropy $H_2(\nu)$. Using the same $P_\delta$ at every context, the shared variable $\lambda=(\lambda_A,\lambda_B)\sim P_\delta$, with Alice returning $\lambda_A$ and Bob returning $\lambda_B$, gives a decomposition of the form of Eq.~\eqref{eq:local} for every admissible $\delta$. In Corollary~\ref{cor:qdt}, dependence between outputs can enter only through $\lambda$, and rules that do not depend on $\lambda$ yield product behaviours for every $\kappa$; the joint correction instead fixes the marginals and changes only the dependence, which shared randomness suffices to produce. Identifying $\delta$ with an attraction or interference factor of a specific decision-theory formulation would require a derivation within that formulation, which we do not supply; the example is an illustrative composite decision model. It also shows that a nonfactorized description need not imply a communication cost, since locality requires only the existence of some decomposition of the form of
Eq.~\eqref{eq:local}.

\subsection{Extended Wigner's friend scenarios}
\label{sec:ewf}

Extended Wigner's friend scenarios~\cite{Wigner1961,Brukner2018,FR2018,Bong2020} declare which agents measure, which records pass to whom, and when. Here the behaviour consists of the superobservers' recorded outputs given their settings.

\paragraph{Observer-independent facts.}
Brukner shows that locality, freedom of choice, and observer-independent facts imply a joint distribution over the outcomes of both settings at each site, which the quantum prediction contradicts through a CHSH violation~\cite{Brukner2018}. In the present terms the joint distribution places the behaviour in $\PC_0$~\cite{Fine1982}, and the result is a separation by Eq.~\eqref{eq:chsh}; see Refs.~\cite{Proietti2019,Healey2018,HC2024} for an experimental test and for later critiques.

\paragraph{Local Friendliness.}
Local Friendliness (LF) combines Absoluteness of Observed Events with Local Agency~\cite{Bong2020,HC2024}. In an LF model the friends' outcomes act like the shared variable $\lambda$ of Eq.~\eqref{eq:local}, and LF behaviours form a polytope between $\PC_0$ and $\NS$ that equals $\PC_0$ in the canonical multipartite scenarios identified in Ref.~\cite{HC2024}. Table~\ref{tab:compare} compares it with the declared classes, and three consequences follow.

\begin{table}[!tb]
\centering
\caption{Comparison of the LF polytope with the declared classes. $\PC_0$ and $\NS$ denote the Bell-local and no-signalling polytopes.}
\label{tab:compare}
\footnotesize
\setlength{\tabcolsep}{3pt}
\renewcommand{\arraystretch}{1.15}
\begin{tabular}{@{}>{\raggedright\arraybackslash}p{0.18\columnwidth}
                >{\raggedright\arraybackslash}p{0.38\columnwidth}
                >{\raggedright\arraybackslash}p{0.38\columnwidth}@{}}
\toprule
 & LF polytope~\cite{HC2024} & $\PC_{\mathcal A}$ (this work)\\
\midrule
Generators
 & Query outputs equal the friends' recorded outcomes; remaining outputs follow a no-signalling extreme point on the non-query contexts~\cite[Thm.~8]{HC2024}
 & Deterministic protocols of the declared schedule (Definition~\ref{def:arch})\\
\addlinespace
Affine equalities
 & Normalization and no-signalling; $\dim\mathrm{LF}=\dim\PC_0=\dim\NS$~\cite[Cor.~12]{HC2024}
 & No-dependence equalities, Eq.~\eqref{eq:nodep}; complete by Theorem~\ref{thm:affine}\\
\addlinespace
Assumptions
 & Absoluteness of Observed Events and Local Agency; space-like separated superobservers; query input reads the friend's record; no communication
 & Fixed acyclic schedule; finite noiseless alphabets; input-independent timing and routing; shared randomness independent of inputs\\
\bottomrule
\end{tabular}
\end{table}

First, every LF polytope shares the affine hull of $\PC_0$ and $\NS$~\cite{Pironio2005,HC2024}, so an LF-violating no-signalling behaviour satisfies every affine equality valid on it and can be certified only by an inequality. A violated no-signalling equality excludes LF through Local Agency and names the remote setting involved, but says nothing about Absoluteness of Observed Events.

Second, the conditional-sampling construction in the proof of Proposition~\ref{prop:two} applies to any finite number of settings when the message alphabet has one symbol per setting of Alice. Every no-signalling behaviour, LF-violating or not, is then reproduced by one message from Alice to Bob; in the scenario with two superobservers and three settings each, where LF inequalities are violated~\cite{Bong2020,HC2024}, one ternary message suffices.

Third, for three superobservers with a friend and two settings each, the LF polytope equals $\PC_0$~\cite{HC2024}, so Mermin's inequality is an LF inequality that the $\sigma_x/\sigma_y$ GHZ behaviour violates when the friends' measurements are modelled unitarily~\cite{HC2024}. The noisy GHZ family is then LF-compatible exactly for $v\le1/2$, and by Proposition~\ref{prop:collapse} no LF inequality of order at most two detects the violation. Because $\PC_0\subseteq\mathrm{LF}$, the behaviours of Corollary~\ref{cor:qdt} also satisfy every LF inequality.

\paragraph{Scenarios without settings.}
The Frauchiger--Renner protocol~\cite{FR2018} involves no setting choices, so its announced results have a single context and are reproduced by shared randomness; none of the present certificates applies.

\paragraph{Possibilistic arguments.}
Hardy-type arguments~\cite{Hardy1992,Hardy1993,HC2023} have a convex counterpart for the declared classes. Write $P(e)$ for $P(a\mid x)$ at $e=(a,x)$. If a target assigns probability zero to every $e$ in a set $Z$ and positive probability to $e^\star$, and every deterministic protocol of the class that assigns zero to $Z$ also assigns zero to $e^\star$, then every protocol satisfies
\begin{equation}
 P(e^\star)\le\sum_{e\in Z}P(e),
 \label{eq:possibilistic}
\end{equation}
because its probabilities are $0$ or $1$, and the target violates this valid inequality. The argument requires deterministic vertices and so does not apply directly to the LF polytope~\cite{HC2024}.

Related reformulations and extensions of LF appear in Refs.~\cite{WCR2023,Ying2024,Walleghem2025}.

\section{Discussion}
\label{sec:disc}

The convex certificate of Eq.~\eqref{eq:membership} is the standard Bell-polytope construction~\cite{Brunner2014}. Communication-assisted polytope methods already distinguish fixed direction, input alphabets, and message budgets~\cite{BaconToner2003,Maxwell2014,Cruzeiro2019}, and the two-agent counts and input-transmission protocol of Sec.~\ref{sec:two} are elementary instances of that framework. Pironio showed that the violation of a Bell inequality lower-bounds the average communication needed to reproduce a behaviour, and that each behaviour has an optimal inequality for which that communication also suffices~\cite{Pironio2003}; average cost is a different resource measure from the fixed event budgets used here. Toner and Bacon address a more extensive measurement family on an entangled pair~\cite{TonerBacon2003}, and Tessier et al. give the $(n-2)$-bit model for Pauli GHZ correlations~\cite{Tessier2005}.

Other generalizations of the local polytope relax different assumptions. Partially deterministic polytopes~\cite{HWC2025}, which include the LF polytopes of Sec.~\ref{sec:ewf}, predetermine the outputs of only some inputs at each site and therefore lie between $\PC_0$ and $\NS$; a declared schedule instead restricts which inputs can reach which outputs and, once an input reaches a remote output, extends beyond $\NS$. Multipartite causal-correlation polytopes allow more general causal orders~\cite{Abbott2016}, and causal-structure analysis distinguishes models whose observable constraints are exhausted by equalities from models that also require inequalities~\cite{Khanna2024}. Theorem~\ref{thm:affine} concerns a narrower model, a fixed acyclic schedule with one shared source, arbitrary local functions, and binary outputs, and provides an explicit equality basis for it; we do not claim that equality constraints or the diagnostic use of no-signalling are new in general. Conditional-independence methods in causal discovery also relate observed distributions to graph structure under stated distributional and faithfulness conditions~\cite{Kalisch2007}. Here the schedule is supplied rather than discovered, and the equalities hold for every admissible distribution without a Gaussian assumption, although independence of the inputs from the shared randomness, timing, retained histories, and the allowed local functions are part of the model. Applied to every $T\subseteq S$, Eq.~\eqref{eq:nodep} gives
\begin{equation}
 P(a_S\mid x)=P(a_S\mid x_{R_{\mathcal A}(S)}),
 \label{eq:marginalnodep}
\end{equation}
the conditional independence of $a_S$ from the unreachable inputs given the reachable ones; a single moment equality generally tests less than this.

The non-obvious content of Theorem~\ref{thm:affine} is Remark~\ref{rem:onebit}: the affine hull depends on reachability alone, so capacity limits, which decide simulability, can be certified only by inequalities. Several questions remain open. Theorem~\ref{thm:affine} determines affine hulls but not facets, and the examples of Secs.~\ref{sec:affine} and~\ref{sec:three} show that capacity constraints can make a declared class strictly smaller than the set of normalized nonnegative behaviours satisfying its equalities. A facet description of declared classes, and the extension of Theorem~\ref{thm:affine} to communication-assisted LF classes in which the space-like separation of Local Agency is replaced by a declared schedule, are left for future work.

\section*{Acknowledgments}
We thank V.~I.~Yukalov for helpful comments.

\bibliographystyle{unsrtnat}
\bibliography{ref}

\clearpage
\appendix
\section{Worked example}
\label{app:example}

This appendix applies the main results to one three-agent task. It illustrates that a violated equality identifies a dependence on an input that cannot reach the agents involved, whereas a target that satisfies every equality, such as one limited only by message capacity, can be excluded only by an inequality.

Alice and Bob hold bits $x$ and $y$, and Charlie's bit $z$ selects which of them matters. With $u=(u_0,u_1)=(x,y)$, the answer and its sign are
\begin{equation}
 u_z=x\oplus z(x\oplus y),\qquad T(x,y,z)=(-1)^{u_z}.
 \label{eq:sel}
\end{equation}
In the \emph{direct} readout Charlie reports the answer: $a=b=0$, $c=u_z$. In the parity readout the group reports it: $a\oplus b\oplus c=u_z$, with $a$ and $b$ uniform and independent. In $P^{\mathrm{dir}}_v$, Charlie's output is replaced with probability $1-v$ by a fresh unbiased bit. In $P^{\mathrm{par}}_v$, the whole behaviour is replaced with probability $1-v$ by $\U$.

We show three things. Equalities exclude the direct readout whenever $x$ or $y$ fails to reach Charlie. The parity readout passes every equality and is excluded only by inequalities. Once all inputs reach Charlie, the size of the last message alone decides the direct readout. Table~\ref{tab:sel} summarizes the results; in the schedule ``Broadcast of $z$ by $C$'', Charlie sends $z$ to both others before any output is produced.

\begin{table}[b]
\centering
\caption{The selection task. $v^\star$ is the largest $v$ for which the noisy target lies in the class. Zero entries are certified by Eq.~\eqref{eq:seleq}, entries $1/2$ by $\mathcal I_{\mathrm{par}}\le4$, and $q/4$ by Eq.~\eqref{eq:idir}, where $m_2$ has $q\in\{2,3,4\}$ symbols.}
\label{tab:sel}
\small
\begin{tabular}{@{}lcccc@{}}
\toprule
Schedule & $R_{\mathcal A}(\{C\})$ & Dim. & $v^\star_{\mathrm{dir}}$ & $v^\star_{\mathrm{par}}$\\
\midrule
No messages & $z$ & 26 & 0 & $1/2$\\
One relay bit & $z$ & 32 & 0 & $1/2$\\
One broadcast by $A$ & $xz$ & 34 & 0 & $1/2$\\
Broadcast of $z$ by $C$ & $z$ & 34 & 0 & 1\\
Two relay messages & $xyz$ & 42 & $q/4$ & 1\\
Two broadcasts & $xyz$ & 44 & 1 & 1\\
\bottomrule
\end{tabular}
\end{table}

Write $A=(-1)^a$, $B=(-1)^b$, and $C=(-1)^c$ for the signed outputs.
Eq.~\eqref{eq:inversion} gives
\begin{equation*}
 P^{\mathrm{dir}}_v=\tfrac18(1+A)(1+B)(1+vTC),\quad
 P^{\mathrm{par}}_v=\tfrac18(1+vT\,ABC).
\end{equation*}
The direct readout has $E_A=E_B=E_{AB}=1$ and $E_S=vT$ for every $S\ni C$. Charlie's marginal therefore depends on $x$ and $y$, and the behaviour signals. The parity readout has $E_{ABC}=vT$ and all other moments zero, so it is no-signalling. Flipping the outputs by a shared uniform sign pattern $\epsilon$ with $\epsilon_A\epsilon_B\epsilon_C=1$ is permitted by assumption~(i) and shared randomness. It removes every moment of $P^{\mathrm{dir}}_v$ except $E_{ABC}$ and yields $P^{\mathrm{par}}_v$. Every class that contains the direct readout therefore contains the parity readout.

The direct readout fails at the equality step. Its moments satisfy
\begin{align}
 E_C(0,0,0)-E_C(1,0,0)&=2v,\nonumber\\
 E_C(0,0,1)-E_C(0,1,1)&=2v.
 \label{eq:seleq}
\end{align}
The first difference violates an equality of Corollary~\ref{cor:eqbasis} whenever $x$ does not reach Charlie, and the second whenever $y$ does not. Each violation names the missing input and persists for every $v>0$, which gives the zero entries of Table~\ref{tab:sel}. The parity readout satisfies every equality of every schedule, so only an inequality can exclude it.

At a local vertex $E_{ABC}=A(x)B(y)C(z)$, where $A(x)$, $B(y)$, $C(z)$ are the signed response tables, and
\begin{equation}
 \mathcal I_{\mathrm{par}}:=\sum_{x,y,z}T\,E_{ABC}
 =C(0)\,\alpha_-\beta_++C(1)\,\alpha_+\beta_-,
 \label{eq:ipar}
\end{equation}
with $\alpha_\pm=A(0)\pm A(1)$ and $\beta_\pm=B(0)\pm B(1)$. One of $\alpha_\pm$ vanishes and the other equals $\pm2$, so $\mathcal I_{\mathrm{par}}\le4$ on $\PC_0$, while $\mathcal I_{\mathrm{par}}(P^{\mathrm{par}}_v)=8v$. Equivalently, every local strategy gives the wrong parity in at least two of the eight contexts, an instance of Eq.~\eqref{eq:possibilistic}.

Guessing $z$ attains the bound. Take shared unbiased bits $g,r_1,r_2,r_3$. If $g=0$, Alice outputs $x\oplus r_1$ and Bob outputs $r_2$; if $g=1$, Alice outputs $r_1$ and Bob outputs $y\oplus r_2$. Charlie outputs $r_1\oplus r_2$ if $z=g$ and $r_3$ otherwise. The parity is $u_z$ when the guess is right and uniform otherwise, and all other moments vanish. Hence $P^{\mathrm{par}}_v\in\PC_0$ exactly for $v\le1/2$.

At $v=1$ the target is also not quantum. Suppose the outputs come from $\pm1$-valued observables $A_x,B_y,C_z$ on separate tensor factors acting on a state $\psi$. Perfect parity at contexts $000$ and $100$ gives $B_0C_0\psi=A_0\psi$ and $A_1B_0C_0\psi=-\psi$, hence $A_1A_0\psi=-\psi$. The same argument at $001$ and $101$ gives $A_1A_0\psi=\psi$, a contradiction. We do not determine the quantum threshold.

Now consider the relay with two messages, where $m_2$ takes $q$ values. Every input reaches Charlie, so the direct readout passes all $14$ equalities for every $q$. Write Charlie's row at $u$ as $\rho(u)=\bigl(c(u,0),c(u,1)\bigr)$. The target requires $\rho(u)=u$, which is the table quoted after Proposition~\ref{prop:three}. Since Charlie sees $u$ only through $m_2$, $\rho$ takes at most $q$ values. Each row contributes $\sum_zT(u,z)(-1)^{c(u,z)}=2-2\,\mathrm d(\rho(u),u)$, where $\mathrm d$ is the Hamming distance, so at most $q$ rows contribute $2$ and the rest at most $0$:
\begin{equation}
 \mathcal I_{\mathrm{dir}}:=\sum_{x,y,z}T\,E_C\le2q,\qquad
 \mathcal I_{\mathrm{dir}}(P^{\mathrm{dir}}_v)=8v.
 \label{eq:idir}
\end{equation}
The bound is attained. In all three protocols Alice sends $m_1=x$ and
Alice and Bob output $0$.
\begin{itemize}
\item $q=2$: for a shared random $j$, Bob sends $u_j$. Charlie outputs it if $z=j$ and a shared unbiased bit otherwise, giving $E_C=T/2$.
\item $q=3$: for shared random $k$ and $t$, Bob sends $u$ unless $u_k=t$, in which case he sends a third symbol. On that symbol Charlie outputs $t$ if $z=k$ and a shared unbiased bit otherwise, giving $E_C=3T/4$.
\item $q=4$: Bob sends $u$.
\end{itemize}
Mixing with $P^{\mathrm{dir}}_0$ covers smaller $v$, so $v^\star_{\mathrm{dir}}=q/4$, although the affine hull is the same for every $q$. The equalities cannot see this difference. The spanning argument of Theorem~\ref{thm:affine} isolates one indicator at a time, each needing one flag bit, whereas the task needs all four rows at once. Two broadcasts deliver $x$ and $y$ separately and realize the direct readout exactly.

The parity readout needs less communication. Alice can absorb the term $x$, which leaves $b\oplus c=z(x\oplus y)$. This is a PR task between Bob, with input $x\oplus y$ once he knows $x$, and Charlie, with input $z$. Without a message between them it obeys the local bound: repeating Eq.~\eqref{eq:ipar} with $B$ depending on $(x,y)$, or with $B$ and $C$ depending on $x$, still gives $4$. One relay bit or one broadcast by Alice therefore does not help. Two relay messages suffice even with binary $m_2$: Alice sends $m_1=x$, Bob sends $m_2=x\oplus y$, and the outputs are $a=r_1$, $b=x\oplus r_1\oplus r_2$, and $c=r_2\oplus z\,m_2$. A broadcast of $z$ by Charlie also suffices, with $a=x\bar z\oplus r_1$, $b=yz\oplus r_2$, and $c=r_1\oplus r_2$, where $\bar z=1\oplus z$.

The weights and inequalities above meet at $v^\star$, so the membership program~\eqref{eq:membership} needs no numerics. Eq.~\eqref{eq:noise} gives $\eta^\star=1/2$ for $\mathcal I_{\mathrm{par}}\le4$ and $\eta^\star=1-q/4$ for Eq.~\eqref{eq:idir}, so both witnesses are optimal along their noise lines. The equalities~\eqref{eq:seleq} instead detect every $v>0$. With $n$ trials per context, Hoeffding's inequality and a union bound give $|\widehat{\mathcal I}-\mathcal I|\le8\sqrt{2\log(16/\alpha)/n}$ with probability at least $1-\alpha$, which yields a level-$\alpha$ test for either inequality.

Low-order statistics miss the parity readout. For every $v$, $\Pi_2(P^{\mathrm{par}}_v)=0$ is the centre of the ball of Eq.~\eqref{eq:ball}, so every valid inequality on moments of order at most two has slack at least $\|\bm c\|_\infty$ there. Like the $\sigma_x/\sigma_y$ GHZ target, the parity readout is detectable only at order three, whereas the direct readout is exposed by the single-agent moment $E_C$.

For agents (Sec.~\ref{sec:agents}), Corollary~\ref{cor:qdt} gives $\mathcal I_{\mathrm{par}}\le4$ for every admissible $\kappa$ when each response uses only its own input and shared prior information; the guessing strategy is such a rule and attains the bound. The marginal-preserving correction of Eq.~\eqref{eq:pdelta}, applied to Alice and Bob, changes only $E_{AB}$, by $4\delta$ at $\nu=1/2$, and so leaves $\mathcal I_{\mathrm{par}}$ unchanged.

For three superobservers with a friend and two settings each, the LF polytope equals $\PC_0$~\cite{HC2024}, so $\mathcal I_{\mathrm{par}}\le4$ is an LF inequality. For $v>1/2$, $P^{\mathrm{par}}_v$ violates it while satisfying every affine equality of the LF polytope. No LF inequality of order at most two detects the violation, and a broadcast of $z$ reproduces it. We do not determine whether quantum theory violates $\mathcal I_{\mathrm{par}}\le4$. With all inputs fixed, as in the Frauchiger--Renner protocol, a single context remains, and shared randomness reproduces both readouts.
\end{document}